\documentclass{ifacconf}

\usepackage{xcolor,soul,framed}
\usepackage{comment} %,caption
\usepackage{natbib}

\colorlet{shadecolor}{yellow}
\usepackage[pdftex]{graphicx}
\graphicspath{{../pdf/}{../jpeg/}}
\DeclareGraphicsExtensions{.pdf,.jpeg,.png}

\usepackage[cmex10]{amsmath}
\usepackage{array}
\usepackage{mdwmath}
\usepackage{mdwtab}
\usepackage{eqparbox}
\usepackage{url}
\AtBeginDocument{\urlstyle{tt}}
\usepackage{bm}
\usepackage{float}
\usepackage{siunitx}
\usepackage[labelformat=simple]{subcaption}
\usepackage{amsmath,amsfonts,amssymb}

\newcommand{\f}{{\mathsf{f}}}

\newcommand{\cJ}{\mathcal{J}}

\newcommand{\cR}{\mathfrak{R}}
\newcommand{\cS}{\mathcal{S}}

\newcommand{\cT}{\mathcal{T}}
\newcommand{\cU}{\mathcal{U}}

\newtheorem{problem}{Problem}
\newtheorem{proposition}{Proposition}
\newenvironment{proof}{\vspace*{-20pt}\paragraph*{Proof:}}{\hfill$\square$}
\begin{document}
\begin{frontmatter}

\title{Vector Field-based Collision Avoidance for Moving Obstacles with Time-Varying Elliptical Shape \thanksref{footnoteinfo}}

\thanks[footnoteinfo]{This research has been supported in part by NSF under ECCS-1924790. }

\author[First]{Martin Braquet}~~~~ 
\author[First]{Efstathios Bakolas} 

\address[First]{Department of Aerospace Engineering and Engineering Mechanics, The University of Texas at Austin, Austin, Texas 78712-1221, USA (e-mail: braquet@utexas.edu, bakolas@austin.utexas.edu).}

\begin{abstract}
This paper presents an algorithm for local motion planning in environments populated by moving elliptical obstacles whose velocity, shape and size are fully known but may change with time. We base the algorithm on a collision avoidance vector field (CAVF) that aims to steer an agent to a desired final state whose motion is described by a double integrator kinematic model. In addition to handling multiple obstacles, the method is applicable in bounded environments for more realistic applications (e.g., motion planning inside a building). We also incorporate a method to deal with agents whose control input is limited so that they safely navigate around the obstacles. To showcase our approach, extensive simulations results are presented in 2D and 3D scenarios.
\end{abstract}


\end{frontmatter}

\maketitle

%\markboth{IEEE TRANSACTIONS ON MICROWAVE THEORY AND TECHNIQUES, VOL.~60, NO.~12, DECEMBER~2012}{Roberg \MakeLowercase{\textit{et al.}}: High-Efficiency Diode and Transistor Rectifiers}

% \fi
%\IEEEpeerreviewmaketitle

% ====================================================================
% ====================================================================
% ====================================================================

% === I. INTRODUCTION =============================================================
% =================================================================================
\section{Introduction}

Collision avoidance constitutes a fundamental problem in robotics.
For example, unmanned aerial vehicles (UAVs) have found many applications related to search and rescue, payload and package delivery, and surveillance, to name but a few. 
%In the multi-agent settings, the agents need to first solve a task allocation problem before considering the motion planing and its underlying obstacle avoidance.
To accomplish their mission, robots have to plan their path in environments that are often populated by obstacles. Such obstacles are not precisely known or can be mobile and hence are often characterized in probabilistic ways (e.g., a confidence or probability ellipsoid). Due to the motion and rotation of these obstacles, as well as the dynamic information gathered by the robot, the probability density of the obstacles (i.e., their shape) may vary with time. Additionally, they are required to react in real time to pop-up tasks and addition or removal of some other agents. Therefore, a well-designed algorithm for such problems needs to be decentralized, reactive, computationally inexpensive, and handle fast collision avoidance.

\textit{Literature review:} 
Some of the earliest and most notable algorithms tackling obstacle avoidance problems consist of creating a graph by discretizing the autonomous agent’s free configuration space, considering only the geometric requirement of finding an obstacle-free path between two points. Some of the most popular graph-based search algorithms are Dijkstra’s (\cite{dijkstra1959note}) and
A* (\cite{hart1968formal,yang2004trajectory}), which require a configuration of the space graph using basic or more advanced sampling techniques such as Voronoi diagrams and projections (\cite{huttenlocher1993upper, aurenhammer1991voronoi}). However, discrete methods rely on search-space granularity (i.e., searching a complex discrete space) to compute the optimal solution, and therefore, are not very suitable for real-time applications.

Many relevant algorithms have been developed for motion planning in dynamic environments using \textit{velocity obstacles}, which aim at selecting avoidance maneuvers to avoid obstacles in the velocity space (\cite{fiorini1998motion}). These algorithms are generated by selecting robot velocities outside of the velocity obstacles, which represent the set of robot velocities that would result in a collision.
Optimal reciprocal collision avoidance (ORCA) is an extension of the velocity obstacle algorithm to include multiple agents by resolving the common oscillation problem in multi-agent navigation (\cite{van2008reciprocal, van2011reciprocal, snape2011hybrid}). 
\begin{comment}
ORCA takes into account the reactive behavior of the other agents by implicitly assuming that they make a similar collision-avoidance reasoning. The agent chooses a velocity that lies outside any of the velocity obstacles induced by the moving obstacles. If among the free velocities, the velocity chosen is the one that is in the most direct way towards the agent’s goal position, the agent will safely navigate towards its goal. ORCA has however two major limitations: deadlocks and local minima (of the objective function representing the path cost) (\cite{trautman2010unfreezing, rufli2013reciprocal}). Indeed in more complex scenarios such as cluttered / obstacle rich environments and moving obstacle environments, the agents can get stuck in local minima where the obstacles block the path to the target or produce poor trajectories as they do not plan for future obstacle locations. Algorithms based on vector fields are able to avoid such local minima by creating a null vector field, and thus a local minimum, only at the desired final location. In this paper, we give a guarantee that the agent reaches the goal and does not get stuck.
To take the \textit{probabilistic} motion into account, an extended velocity obstacle (EVO) validation system extends ORCA in the real world by including state uncertainties via a position filtering system which handles a noise model that is discontinuous and non-linear (\cite{levy2015extended}).
\end{comment}

More recently, new algorithms have been leveraging the qualities of machine learning to tackle problems in unknown / sparse dynamic environments. For example, an algorithm for socially aware multiagent collision avoidance with deep reinforcement learning (SA-CADRL) has been designed to produce more time-efficient paths than ORCA (\cite{chen2017socially, long2018towards}).

Algorithms based on harmonic potential flow can deal with multiple moving convex and star-shaped concave obstacles by applying a dynamic modulation matrix, dependent on the characteristics of the obstacles, to the original dynamic system (\cite{khansari2012dynamical, huber2019avoidance}). However, these velocity-based controllers do
not take into account the inertia of the robot or more realisitc kino-dynamic constraints.

Local methods represent a robot as a particle in the configuration space under the influence of an artificial potential field, whereas the direction of motion is generally chosen according to the gradient of the potential function (\cite{Latombe1991,warren1990multiple}). However, such methods suffer from important limitations: trap situations due to local minima (cyclic behavior), no passage between closely spaced obstacles, oscillations in the presence of multiple obstacles or narrow passages (\cite{koren1991potential}).

\begin{comment}
Global path planning methods such as rapidly-exploring random trees (RRT) have the advantage that they can be directly applied to nonholonomic and kinodynamic planning (\cite{lavalle1998rapidly}). Indeed, RRT generates a tree by providing the required control input making the link between two adjacent states. General motion planning between two locations can be used more efficiently than the classical RRT by incrementally building two rapidly-exploring random trees (RRTs) rooted at the start and the goal configurations (\cite{kuffner2000rrt, lavalle2001rapidly, goerzen2010survey}).
The previously mentioned field-based algorithms (harmonic flow, potential, etc) can further be combined with such global path planning to solve the problem of artificial potential field algorithms being liable to fall into a a local minimum (\cite{he2017obstacle}).

Other popular approaches for collision avoidance rely on curve parameterization for trajectory generation (\cite{upadhyay2017smooth}) or optimization-based algorithms that aim to solve nonlinear programs with different nonlinear programming (NLP) solvers (\cite{sun2017two}). Still, these methods are computationally expensive and not suited for real-time applications.
\end{comment}

To avoid the computational complexity of the previously discussed techniques, algorithms based on collision avoidance vector fields have been developed for specific types of dynamics. For example, \cite{marchidan2020collision} considered UAV applications by assuming constant-altitude operations and thus approximating the system with a planar kinematic Dubins model. In this case, the agent aims to avoid circular obstacles while keeping its steering angle close to a desired angle. \cite{panagou2014motion, panagou2016distributed} developed a motion planning method for nonholonomic (unicycle) robots in environments with static circular obstacles.
\cite{wilhelm2019vector} considered a gradient vector field (GVF) for real-time guidance systems and used lookup tables for optimal GVF decay radius and circulation. \cite{garg2019finite} also developed trajectories for agents under double-integrator dynamics with limited erroneous sensing capabilities in the presence of unknown wind disturbance and moving obstacles. However, the applicability of the previously mentioned methods are limited to circular obstacles whose shape does not vary with time. 

\textit{Contributions:} This paper presents a novel local motion-planning algorithm in environments populated by \textit{multiple moving obstacles with time-varying elliptical shapes}. We base our algorithm on a collision avoidance vector field (CAVF) that aims to steer an agent to a desired final state under double integrator dynamics.

We first assume that the agent has access to the exact position and velocity of the center of mass of the obstacles, as well as their shape and orientation. Secondly we consider obstacles with uncertain motion resulting in a position covariance that is growing over time. We also consider bounded environments for more realistic applications. Finally, we incorporate a method to deal with agents whose control input is limited so that they can safely navigate around the obstacles.

%Multi-agent systems (MAS) further require a task assignment scheme in addition to the avoidance of obstacles and other agents. This work is extended to multi-agent settings by incorporating a fast task allocation algorithm wrapped around our proposed collision avoidance method (\cite{braquet2021greedy}).

\textit{Outline}: The remainder of the paper is structured as follows. Section \ref{sec:prob_setup} presents the problem setup, Section \ref{sec:cavf_design} presents the core of our algorithm and Section \ref{sec:control} details the control law derived from the vector field. In addition, Section \ref{sec:simu} shows extensive simulations results and finally, Section \ref{sec:ccl} concludes this work.

\section{Problem Setup} \label{sec:prob_setup}

We consider an agent moving in an $N$-dimensional position space $\Omega \subseteq \mathbb{R}^N$ (where $N=2$ or $N=3$) which can be a bounded set. The state of the agent is defined as 
\[
    \bm{x} = \begin{pmatrix} \bm{P} \\ \bm{V} \end{pmatrix}
\]
where $\bm{P} \in \Omega$ is the position of the agent and $\bm{V} \in \mathbb{R}^N$ the velocity.
We consider an obstacle whose center of mass is initially located (at time $t=0$) at $\bm{P}_o \in \Omega$ and has a constant velocity $\bm{V}_o \in \mathbb{R}^N$, which implies that the obstacle's motion is reduced to a straight line if $\bm{V}_o \neq 0$ (if $\bm{V}_o = 0$, the obstacle is static). 

We also assume that the obstacle's boundary $\chi^b$ can be parametrically described as a function of two variables $u \in [0,2\pi]$ and $v \in [0,\pi]$. 
In the case in which the obstacle's boundary is an ellipse (2D), we have
\[
    \bm{P}^b(t;u) = \begin{pmatrix} a(t)\cos(u) \\ b(t)\sin(u) \end{pmatrix},
\]
whereas in the case of an ellipsoid (3D), we have
\[
    \bm{P}^b(t;u,v) = \begin{pmatrix} a(t)\cos(u)\sin(v) \\ b(t)\sin(u)\sin(v) \\ c(t)\cos(v) \end{pmatrix}.
\]
The velocity $\bm{V}_e^{\bm{P}}$ due to expansion/contraction at a point $\bm{P}^b \in \chi^b$, corresponding to a pair $(u,v)$, can be described as
\[
    \bm{V}_e^{\bm{P}} = \frac{ \partial \bm{P}^b(t;u,v) }{\partial t }.
    \label{eq:V_exp}
\]

The agent is subject to a double integrator dynamics:
\begin{equation} \label{eq:dyn}
    \dot{\bm{x}}(t) = \begin{pmatrix} \bm{V}(t) \\ \bm{u}(t) \end{pmatrix}, \quad \bm{x}(0) = \bm{x}_0 = \begin{pmatrix} \bm{P}_0 \\ \bm{V}_0 \end{pmatrix}
\end{equation}
where $\bm{u}(t) \in \mathbb{R}^N$ is the control input. %It is worth noting that, assuming a bounded input control, this dynamics takes into account the inertia of the robot.

Next, we formulate the main problem considered in this paper.

\begin{problem} \label{prob_init}
Let an agent be subject to the dynamics given in Eq. (\ref{eq:dyn}). Then, find a control input $\bm{u}(t) : [0,t_f] \rightarrow \mathbb{R}^N$ so that the agent will reach a final desired location $\bm{P}_f \in \Omega$ with zero velocity at a free final time $t_f$:
\[
    \bm{x}(t_f) = \bm{x}_f = \begin{pmatrix} \bm{P}_f \\ \bm{0} \end{pmatrix},
\]
while avoiding collisions with obstacles.
\end{problem}

\section{CAFV design} \label{sec:cavf_design}

The collision avoidance vector field (CAVF) is designed to match the desired agent's velocity along the state space. For example, the CAFV is directed towards the final location and is null at the final location. Because of obstacles appearing on the way from the agent's location to its final location, the CAFV has to be modulated so that the agent does not enter into collision (see Fig. \ref{fig:CAVF_design}).

\begin{figure}[ht!]
\centering
\begin{subfigure}{0.49\linewidth}
\includegraphics[width=\linewidth]{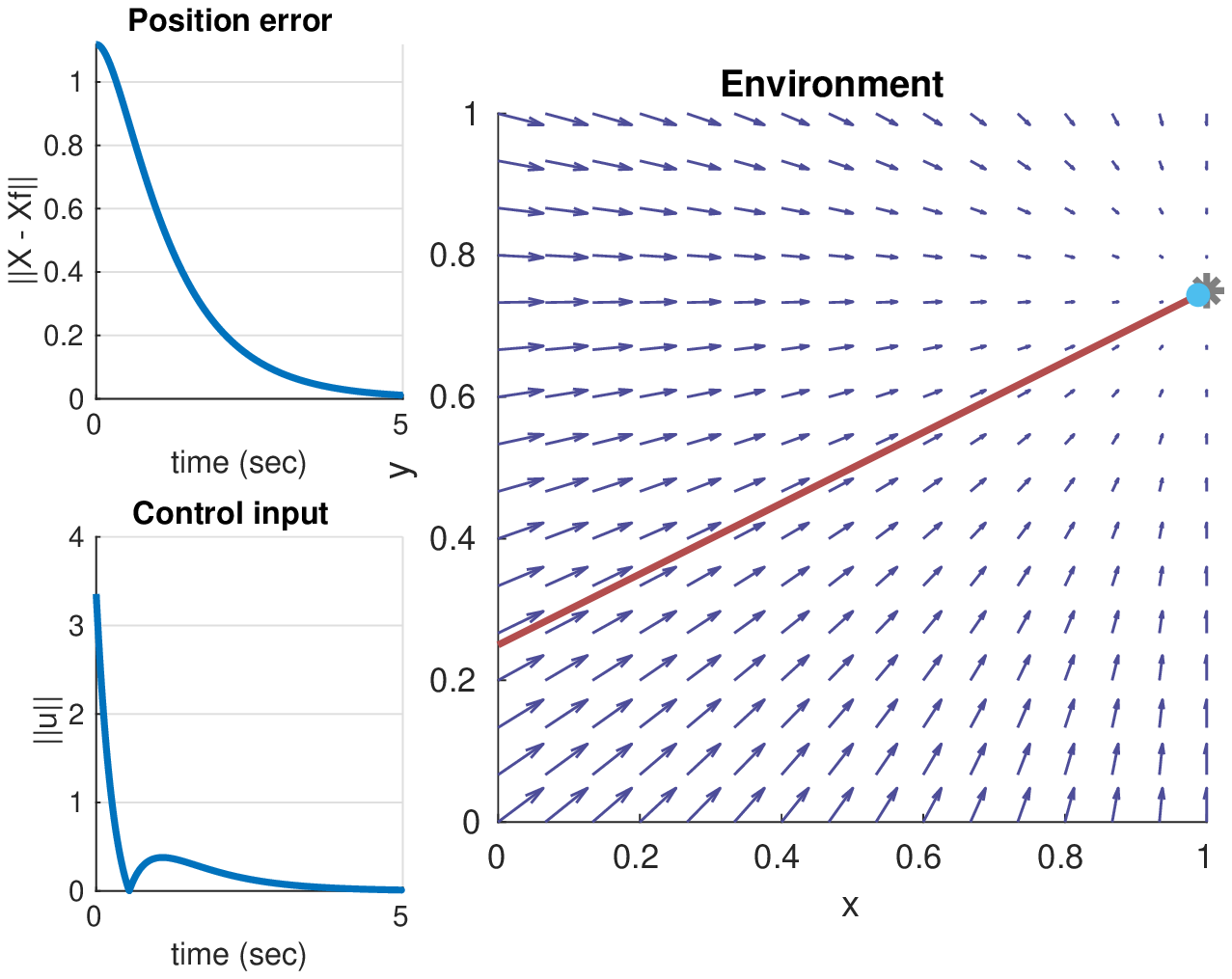}
\caption{Without obstacle.}
\end{subfigure}
\begin{subfigure}{0.49\linewidth}
\includegraphics[width=\linewidth]{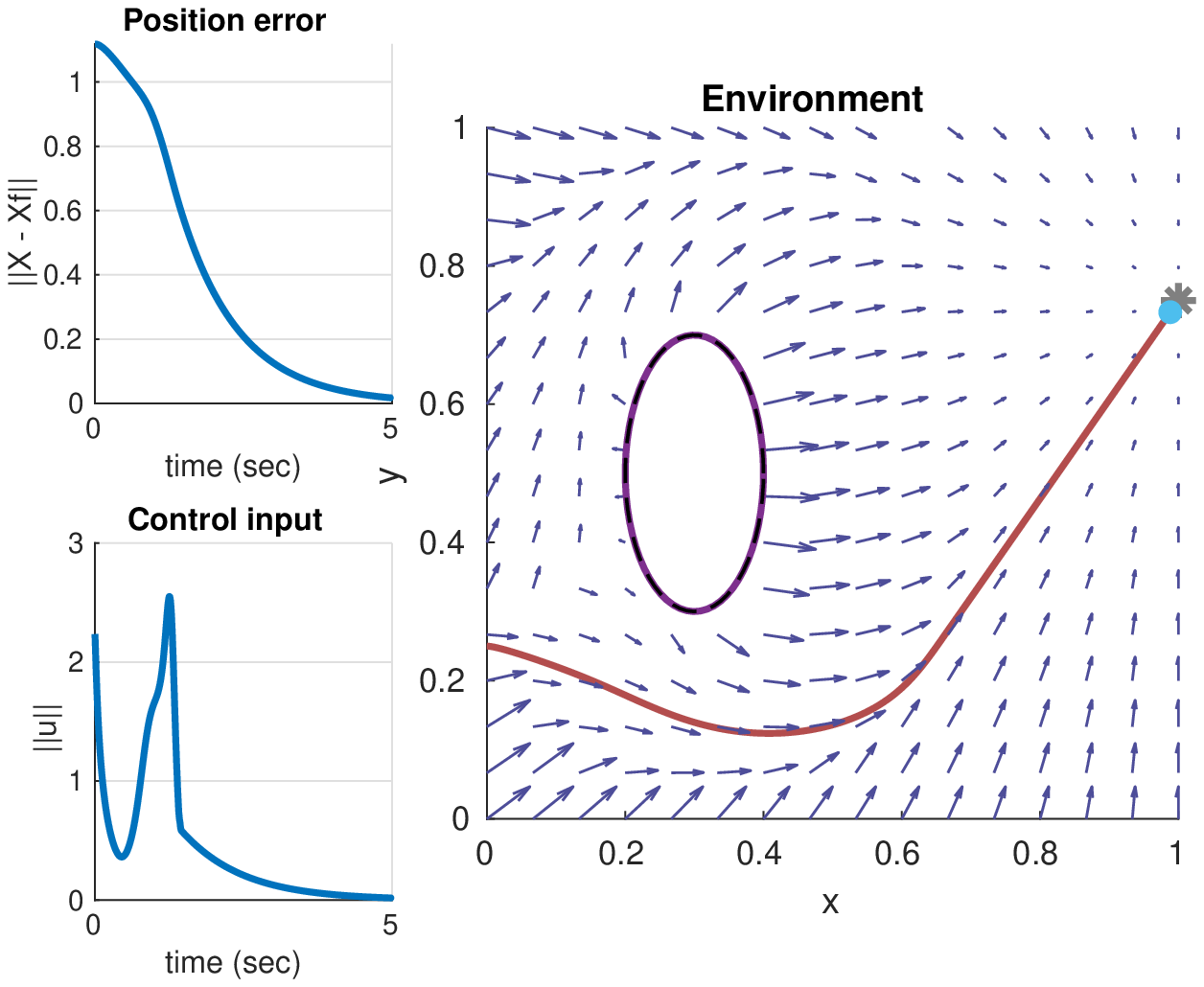}
\caption{With obstacle.}
\end{subfigure}
\caption{In the absence of obstacles (a), the CAFV points towards $\bm{P}_f$. In the presence of obstacles (b), the CAFV is modulated to avoid the obstacle ($\bm{P}_0 = [0,0.25]^\top$ and $\bm{P}_f = [1,0.75]^\top$).}
\label{fig:CAVF_design}
\end{figure}

\subsection{Static Obstacle}

For a static obstacle, the collision avoidance vector field $\bm{h}$ has to satisfy the following impenetrability condition at any point $\bm{P}^b$ of the obstacle boundary:
\begin{equation}
\label{eq:stat_bos}
\bm{h}(\bm{P}^b) \cdot \hat{\bm{n}}(\bm{P}^b) \ge 0 \qquad \textrm{for all } \bm{P}^b \in \chi^b,
\end{equation}
where $\hat{\bm{n}}(\bm{P})$ is the unit vector that is normal to the obstacle boundary (pointing outwards).
\begin{comment}
and is given by
\[
    \hat{\bm{n}}(\bm{P}) = \left\{ \begin{matrix}
     \frac{\bm{P}-\bm{P}^c(\bm{P})}{\|\bm{P}-\bm{P}^c(\bm{P})\| } & \textrm{if } \bm{P} \not\in \chi^b \\
     \frac{\bm{P}_l(\bm{P})-\bm{P}}{\|\bm{P}_l(\bm{P})-\bm{P}\| } & \textrm{if } \bm{P} \in \chi^b
    \end{matrix} \right.  
\]
where , and $\bm{P}^l(\bm{P})$ is any point satisfying $\bm{P}^c(\bm{P}^l(\bm{P})) = \bm{P}$ (i.e., any point on the line determined by $\bm{P} + \alpha \, \hat{\bm{n}}(\bm{P})$ for all $\alpha > 0$). 
\end{comment}
We define the vector field $\bm{h}$ as the superposition of an obstacle avoidance field $\bm{h}_o$ and a destination steering field $\bm{h}_f$ (to reach the aimed location); that is, $\bm{h}(\bm{P})=\bm{h}_o(\bm{P}) + \bm{h}_f(\bm{P})$:
\begin{align*}
   \bm{h}_o(\bm{P}) &= \bm{R}(\bm{P}) \ \|\bm{P}_f - \bm{P}\|^{-p} \ \gamma(\bm{P}) \, \|\bm{P}_f - \bm{P}\| \hat{\bm{n}}(\bm{P}^c(\bm{P})), \\
    \bm{h}_f(\bm{P}) &= \bm{R}(\bm{P}) \ \|\bm{P}_f - \bm{P}\|^{-p} \ (\bm{P}_f - \bm{P})
\end{align*}
and hence
\begin{align} \label{eq:h}
    &\bm{h}(\bm{P}) = \\
    &\bm{R}(\bm{P}) \ \|\bm{P}_f - \bm{P}\|^{-p} \Big[ \gamma(\bm{P}) \, \|\bm{P}_f - \bm{P}\| \hat{\bm{n}}(\bm{P}^c(\bm{P})) + (\bm{P}_f - \bm{P}) \Big] \nonumber
\end{align}
where $p \in ]0,1[$ is a modulation exponent which controls the difference of intensity across the vector field. $\bm{P}^c(\bm{P})$ is defined as the point on the obstacle boundary that is closest to $\bm{P}$, which can be computed via an iterative optimization algorithm (\cite{nurnberg2006distance}). We introduce a location-dependent factor $\gamma \in [0,1]$ which is defined as follows:
\[
    \gamma(\bm{P}) = \frac{a_i \ x(\bm{P})}{\sqrt{1+(2 \ a_i \ x(\bm{P}))^2}} + \frac{1}{2}
\]
where $x(\bm{P}) := (d(\bm{P})+d_2(\bm{P})) / (d(\bm{P}) \, d_2(\bm{P}))$, $d_i$ is the influence distance of the obstacle (typically in the order of the obstacle semi-major axis), $d(\bm{P}) = \|\bm{P}-\bm{P}^c(\bm{P})\|$ and $d_2(\bm{P}) = d(\bm{P})-d_i$. In addition, $a_i$ is a modulation constant to modify the shape of the sigmoid given by $\gamma$.
The vector field $\bm{h}$ guarantees that $\bm{h}_o = \bm{0}$ when the agent is located at the influence distance since $d=d_i$ implies $x \rightarrow -\infty$ and hence $\gamma = 0$.
Furthermore, a rotation operator $\bm{R}$ of angle $\alpha_R$ is applied to the vector field so that the agent goes around the obstacle when the obstacle is on the line of sight between the agent and the target position. 
More formally, for $N=2$, the matrix $\bm{R}$ at a point $\bm{P}$ is defined as 
\[
    \bm{R}(\bm{P}) = 
    \begin{bmatrix}
     \cos(\alpha_R(\bm{P})) & \sin(\alpha_R(\bm{P})) \\
     - \sin(\alpha_R(\bm{P})) & \cos(\alpha_R(\bm{P})) \\
    \end{bmatrix}.
\]
Let us define $\beta(\bm{P}) = e^{-b_i x^2(\bm{P})}$ where $b_i \in \mathbb{R}$ is a fixed parameter. The rotation angle $\alpha_R$ is then computed proportionally to the angle between $\hat{\bm{n}}$ and $\bm{P}_f-\bm{P}_o$:
\[
    \alpha_R(\bm{P}) = \frac{\beta(\bm{P})}{2} \angle ( \hat{\bm{n}}(\bm{P}), \bm{P}_f-\bm{P}_o)
\]
where $\bm{P}_f-\bm{P}_o$ corresponds to the vector pointing from the obstacle center $\bm{P}_o$ to the final position  $\bm{P}_f$. The factor $\beta$ is introduced to ensure that the vector field is not rotated (that is, $\beta = 0$ and hence $\bm{R} = \bm{I}_N$ where $\bm{I}_N$ is the identity matrix of order $N$) when the agent is located at the influence distance of the obstacle (continuity of the vector field), and when the agent is located on the obstacle boundary (non-penetrability of the obstacle).
%In 3D, $\bm{R}$ is defined as the rotation matrix in the plane determined by the vectors $\hat{\bm{n}}$ and $\bm{P}_f-\bm{P}_o$.

Now that the vector field has been described analytically, we can analyze two particular cases. The first scenario consists in the vectors $\hat{\bm{n}}$ and $\bm{P}_f-\bm{P}_o$ making an angle of $\pi$ (singular case), corresponding to a point $\bm{P}$ exactly behind the obstacle with respect to the desired location. In this case, the field is rotated by an angle $\alpha \in [-\pi/2,\pi/2]$ so that the agent deviates from its initial trajectory to go around the obstacle (forming a singularity in the vector field). The second scenario appears when these two vectors $\hat{\bm{n}}$ and $\bm{P}_f-\bm{P}_o$ are aligned, that is, when the point $\bm{P}$ lies in between the obstacle and the target. In this case, the vector field is not rotated ($\alpha = 0$) since the agent is in front of the obstacle and can freely move to the target. The vector field is illustrated in Fig.~\ref{fig:path_sing} where the target position is represented by a red cross.

\begin{figure}[h]
  \centering
  \includegraphics[width=0.8\linewidth]{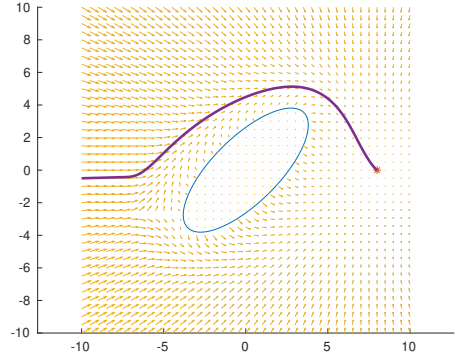}
  \caption{Vector field and associated path line for an initial agents location $\bm{P}_0 = [-10,0]^\top$. The singularity appears around $[-5,-3]^\top$ (behind the obstacle).}
  \label{fig:path_sing}
\end{figure}

\begin{proposition} \label{prop:imp_stat}
With the vector field defined in Eq. (\ref{eq:h}), the condition of impenetrability given in Eq. (\ref{eq:stat_bos}) is guaranteed for the case of a static obstacle.
\end{proposition}
\begin{proof}
When the agent is located at a point $\bm{P}^b \in \chi^b$ (on the boundary of the obstacle), we obtain $\bm{P}^c(\bm{P}^b) = \bm{P}^b$, $d = 0$, $x  \rightarrow \infty$ and hence $\gamma = 1$. In addition, we have $\beta = 0$, $\alpha_R = 0$ and thus $\bm{R} = \bm{I}_N$. Given the particular value of the vector field
\[
    \bm{h}(\bm{P}^b) =  \|\bm{P}_f - \bm{P}\|^{-p} \Big[ \|\bm{P}_f - \bm{P}\| \hat{\bm{n}}(\bm{P}^b) + (\bm{P}_f - \bm{P}) \Big],
\]
we obtain
\begin{align*}
    & \bm{h}(\bm{P}^b) \cdot \hat{\bm{n}}(\bm{P}^b)  \\ &~~~=\|\bm{P}_f - \bm{P}\|^{-p} \Big[ \|\bm{P}_f - \bm{P}\| + (\bm{P}_f - \bm{P}) \cdot \hat{\bm{n}}(\bm{P}^b) \Big] \\
    &~~~=  \|\bm{P}_f - \bm{P}\|^{-p} \Big[ \|\bm{P}_f - \bm{P}\| + \underbrace{(\bm{P}_f - \bm{P}) \cdot \hat{\bm{n}}(\bm{P}^b)}_{\ge -\|\bm{P}_f - \bm{P}\|} \Big] \ge 0,
\end{align*}
which verifies Eq. (\ref{eq:stat_bos}) of impenetrability and thus concludes the proof.
\end{proof}

\subsection{Moving Obstacle}

The velocity $\bm{V}_b$ of the obstacle boundary at $\bm{P}^b$ can be the result of multiple contributing terms. This paper will examine 2 movements: constant velocity translation ($\bm{V}_o$) and expansion/contraction of the boundary in time ($\bm{V}_e^{\bm{P}}$) so that $\bm{V}_b = \bm{V}_o + \bm{V}_e^{\bm{P}}$.
Additionally, we consider the obstacle’s velocity in the normal direction of the obstacle boundary only if it is greater than zero ($\bm{V}_b \cdot \hat{\bm{n}} > 0$) so that the agent is not attracted to the obstacle when it is moving away from it.

To account for the fact that the obstacle is moving, we introduce an additional term $\bm{h}_v(\bm{P})= \gamma \bm{V}_b$ to the vector field such that the total vector field $\bm{h}_m(\bm{P})$ for a moving obstacle is given by
\begin{equation} \label{eq:hmov}
    \bm{h}_m(\bm{P}) = \bm{h}(\bm{P}) + \bm{h}_v(\bm{P}).
\end{equation}
When the obstacles are moving, the impenetrability condition requires that the component of the vector field normal to the boundary has to satisfy the following inequality: 
\begin{equation}
\label{eq:mov_bos}
\bm{h}_m(\bm{P}^b) \cdot \hat{\bm{n}}(\bm{P}^b) \ge \bm{V}_b \cdot \hat{\bm{n}}(\bm{P}^b) \qquad \textrm{for all } \bm{P}^b \in \chi^b.
\end{equation}

\begin{proposition} \label{prop:imp_mov}
With the vector field defined in Eq. (\ref{eq:hmov}), the condition of impenetrability given in Eq. (\ref{eq:mov_bos}) is guaranteed for an obstacle moving at constant velocity $\bm{V}_b$.
\end{proposition}
\begin{proof}
When the agent is located at the boundary $\bm{P}^b$ of the obstacle, we have $\gamma = 1$. It follows that
\begin{align*}
    \bm{h}_m(\bm{P}^b) \cdot \hat{\bm{n}}(\bm{P}^b) &=  \underbrace{\bm{h}(\bm{P}^b) \cdot \hat{\bm{n}}(\bm{P}^b)}_{\ge 0 \textrm{ from Eq. (\ref{eq:stat_bos}})} + \bm{V}_b \cdot \hat{\bm{n}}(\bm{P}^b)  \\
    & \ge \bm{V}_b \cdot \hat{\bm{n}}(\bm{P}^b),
\end{align*}
which verifies Eq. (\ref{eq:mov_bos}) (condition of impenetrability) and thus concludes the proof.
\end{proof}

In addition, $\bm{h}_v(\bm{P}) = \bm{0}$ when the agent is located at the influence distance ($d=d_i$) so that $\bm{h}(\bm{P}) = \bm{h}_f(\bm{P})$ as if there were no obstacles.

\begin{comment}

\subsection{Uncertain moving obstacle}

We now consider an obstacle whose state is not precisely known by the agent. To this end, the state $\bm{x}_o$ of the obstacle is represented with a Gaussian density function characterized by its mean $\bar{\bm{x}}_o$ and covariance matrix $\bm{P}_{o}$. When the obstacle is moving, the uncertainty (i.e., the covariance matrix) is growing with time. 

To reduce this uncertainty, we consider the case where the agent has access to local information through sensors. In this case, such additional information is decreasing the uncertainty about the state of the obstacle as the agent progresses close to this obstacles.

To illustrate the principle, we consider a double integrator for the motion of the obstacle. We also assume that the agent has access to a (noisy) sensor providing the distance to the obstacle in polar coordinates with respect to the agent.

An extended Kalman filter is then used to integrate prediction and sensing. Prediction is achieved via the obstacle dynamics to predict the state estimate of the obstacle over time as it is moving in the environment. Sensing is then used to merge the local data from the agent's sensors and hence update the state estimate.

\end{comment}

\subsection{Multiple Obstacles}

Next, we consider the case of multiple (non-overlapping) obstacles.
Let $d_j = \|\bm{P}-\bm{P}^c_j(\bm{P})\|$ be the distance between $\bm{P}$ and $\bm{P}^c_j(\bm{P})$, where $\bm{P}^c_j(\bm{P})$ is the point on the boundary of obstacle $j$ that is closest to $\bm{P}$.
When there are $M$ obstacles, the vector field is extended by taking the sum of the local CAVFs weighted by their distance to the agent's position $\bm{P}$ so that the vector field at the boundary of obstacle $i$ is only determined by the local vector field $\bm{h}_i(\bm{P})$ associated to this obstacle: 
\begin{equation} \label{eq:hmult}
    \bm{h}(\bm{P}) = \sum_{i=1}^M w_i(\bm{P}) \, \bm{h}_i(\bm{P}),
\end{equation}
where
\[
    w_i(\bm{P}) = \frac{\prod_{j\neq i}^M d_j}{\sum_{i=1}^M \prod_{j\neq i}^M d_j} \quad \textrm{and hence} \quad \sum_{i=1}^M w_i(\bm{P}) = 1.
\]
\begin{proposition}
With the vector field defined in Eq. (\ref{eq:hmult}), the condition of impenetrability given in Eq. (\ref{eq:stat_bos}) is guaranteed for the case of multiple (non-overlapping) obstacles.
\end{proposition}
\begin{proof}
When the agent is located at the boundary $\bm{P}^b_j$ of obstacle $j$, we obtain $w_i(\bm{P}^b_j) = 1$ for $i = j$ and $w_i(\bm{P}^b_j) = 0$ for $i \neq j$. The resulting vector field is thus given by $\bm{h}(\bm{P}^b_j) = \bm{h}_j(\bm{P}^b_j)$. The impenetrability condition for multiple obstacles is therefore guaranteed since the impenetrability condition at the boundary of a single obstacle has already been proven in Proposition \ref{prop:imp_stat}.
\end{proof}

\subsection{Bounded Environment}

We now consider an environment that is a convex polygonal set such that the agent is not allowed to exit through its boundaries (for instance $\Omega$ is a rectangular domain in 2D or a rectangular parallelepiped in 3D). This extension is implemented by assigning a flat ellipsoidal obstacle to each edge. More formally, for an ellipse whose semi-axis lengths are given by $a$ and $b$, we consider a large semi-major over semi-minor length ratio: $a/b \gg 1$. 

\section{Control Law} \label{sec:control}

As mentioned earlier, we consider a double integrator dynamics given in Eq. (\ref{eq:dyn}).
% \[    \ddot{\bm{P}}(t) = \bm{u}(t). \]
Since the vector field $\bm{h}$ reflects the desired agent's velocity $\bm{V}(t)$, we aim to find a control law providing a link between $\dot{\bm{h}}$ and $\dot{\bm{V}}(t) = \bm{u}(t)$. We therefore consider a control law mixing the proportional gain and the derivative of the vector field:
\[
\bm{u} = k_p (\bm{h}-\bm{V}) + k_v \dot{\bm{h}}
\]
where $k_p \in \mathbb{R}$ is the proportional gain, $k_v \in \mathbb{R}$ is the derivative gain and $\dot{\bm{h}} = \nabla \bm{h} \cdot \bm{V}$ following the chain rule.
The agent will only follow the CAVF (velocity vectors) if the initial velocity matches the velocity of the vector field at the initial position. For example when $\bm{V} = \bm{0}$ at start, we have $\dot{\bm{h}} = \nabla \bm{h} \cdot \bm{V} = \bm{0}$ as well. If we use the simple control law $\bm{u} = \dot{\bm{h}}$, we would obtain $\bm{u} = \bm{0}$ and the agent would never start. We thus also need a proportional controller at start to launch the agent and make it match the desired velocity. Once the desired speed is similar to the local vector field ($\bm{h} \simeq \bm{V}$), the second term of the proposed control law based on $\dot{\bm{h}}$ takes the lead, taking into account not only the value of $\bm{h}$ at the agent’s location but also the variation $\nabla \bm{h}$ of the vector field around this location to better predict the best trajectory. This gradient can be computed numerically via the symmetric difference quotient.

\subsection{Norm-Constrained Control Input}

When the input is constrained, the agent might not be able to escape an obstacle if its speed is too high near the obstacle. We consider the case in which $\bm{u} \in \mathcal{U} = \{\bm{u} \in \mathbb{R}^N : \|\bm{u}\| \le \bar{u} \}$ where $\bar{u}\in \mathbb{R}^+$. The time to stop the agent with initial velocity $\bm{V}_0$ is then $t = \|\bm{V}_0\| / \bar{u}$. Based on the motion equation $x(t) = - \bar{u}t^2/2 + V_0t$, the distance that has been traveled in this time is given by $$R = \frac{V_0^2}{2\bar{u}}$$, where $V_0$ is the radial velocity of the agent towards the obstacle.
We artificially inflate each obstacle by the distance $R$ so that the agent can adjust its trajectory to avoid lying inside the inflated boundary of the obstacle. It is worth noting that this technique which aims at avoiding collision, it does not guarantee impenetrability (neither of the inflated obstacle nor of the real obstacle) in all conditions and further research is needed. 

When the obstacle is moving, the velocity of the agent in the obstacle frame is given by $V_0 - V_b$ where $V_b$ is the velocity at the obstacle boundary along the line of sight towards the agent (composed of the obstacle velocity and the expansion velocity).
%If the agent lies inside this inflated boundary, obstacle avoidance is still possible but it is not guaranteed. 

\section{Numerical Simulations} \label{sec:simu}

The following simulations are computed in 2D and 3D scenarios for a system with double integrator dynamics. The agent is considered to be a particle (point mass) and all the simulations have been realized with the following parameters: $d_i = 0.3, a_i = 0.01$ and $p=0.5$. The video presenting all the cases presented in this section is available at \url{https://youtu.be/xSIyvFngA1M}. The source code used to produce the paper is available at \url{https://github.com/MartinBraquet/vector-field-obstacle-avoidance}.

\subsection{Multiple Static Obstacles}

The CAVF design for 2D static obstacles has been presented in Section~\ref{sec:cavf_design} (see Fig. \ref{fig:CAVF_design}). Fig. \ref{fig:multiellipsoid} presents the motion of an agent in a 3D space, moving around multiple ellipsoids and reaching the desired final location.

\begin{figure}[H]
  \centering
  \includegraphics[width=1\linewidth]{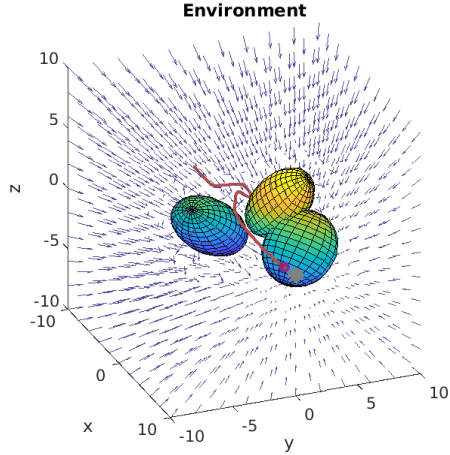}
  \caption{3D obstacle avoidance with multiple ellipsoids.}
  \label{fig:multiellipsoid}
\end{figure}

\begin{comment}
\begin{figure}[H]
  \centering
  \includegraphics[width=\linewidth]{img/ellipsoid.eps}
  \caption{3D obstacle avoidance with ellipsoid ($\bm{P}_0 = [-10,10,0]^\top$ and $\bm{P}_f = [10,-10,0]^\top$).}
  \label{fig:ellipsoid}
\end{figure}
\end{comment}

\subsection{Moving Obstacles with Time-Varying Shape}

In Fig. \ref{fig:movingellipse}, the elliptical obstacle is located in the bottom-right quadrant at start ($\bm{P}_o = [0.6,0.2]^\top$) and is moving upwards at constant velocity (dashed lines).

\begin{figure}[H]
  \centering
  \includegraphics[width=1\linewidth]{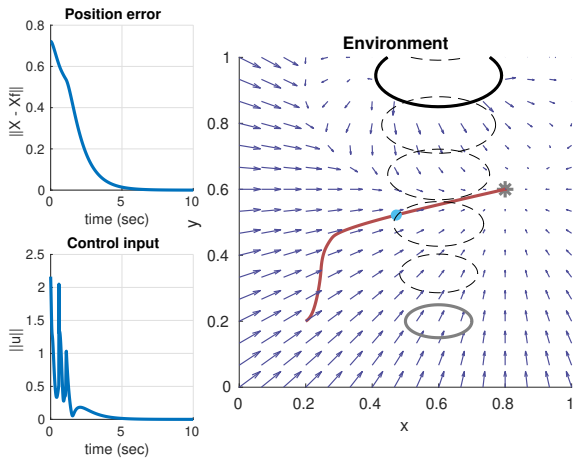}
  \caption{2D obstacle avoidance with a moving ellipse of time-varying shape represented in plain gray (initial position), dashed black (intermediate positions) and plain black (final position) ($\bm{P}_0 = [0.2,0.2]^\top$ and $\bm{P}_f = [0.8,0.6]^\top$).}
  \label{fig:movingellipse}
\end{figure}

Additionally, moving obstacles lead to growing uncertainty in probabilistic motion. We consider scenarios where such uncertainty is represented by an ellipse of time-varying shape (growing with time, see video whose link is given at the start of the section for further illustrations). The agent starts avoiding the obstacle by going above the ellipse. Since the agent is moving with a similar speed as the obstacle, he then decides to slow down and let the obstacle move upwards so that he reaches his target location by going below the obstacle. 

The norm of the position error, given by $\|e(t)\| = \|\bm{P}(t) - \bm{P}_f\|$, is decreasing exponentially.
%(as discussed in Section \ref{sec:convergence}).
The norm of the control input $u(t)$, which is assumed unbounded in this simulation is high at the beginning since the initial zero velocity does not match the desired velocity of the vector field at the initial agent's location. 

\begin{comment}
\subsection{Multiple Obstacles}

\begin{figure}[H]
  \centering
  \includegraphics[width=1\linewidth]{img/multiellipse.eps}
  \caption{2D obstacle avoidance with multiple ellipses ($\bm{P}_0 = [0.4,0.05]^\top$ and $\bm{P}_f = [0.6,0.8]^\top$).}
  \label{fig:multiellipse}
\end{figure}
\end{comment}

\subsection{Bounded Environment}

As shown in Fig. \ref{fig:walls}, the environment is a square domain, that is, bounded by four walls, and an ellipse is located at the center. The vector field is directed outward around the ellipse and inward next to the walls, the agent is thus guaranteed to stay in the desired area and he is correctly moving towards his target.

\begin{figure}[H]
  \centering
  \includegraphics[width=1\linewidth]{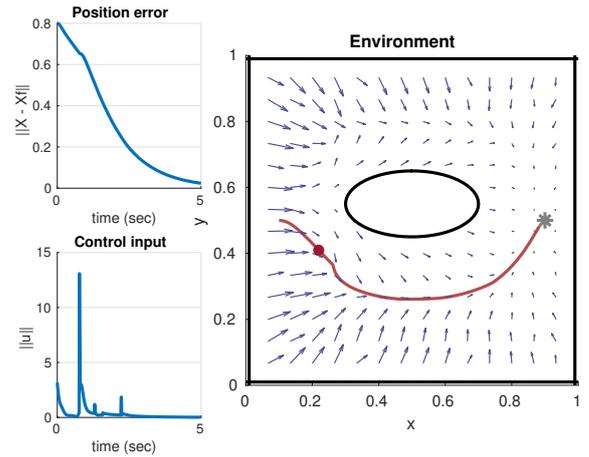}
  \caption{2D obstacle avoidance with an ellipse in a bounded environment ($\bm{P}_0 = [0.1,0.5]^\top$ and $\bm{P}_f = [0.9,0.5]^\top$).}
  \label{fig:walls}
\end{figure}

\subsection{Norm-Constrained Input Control}

We finally consider the case where the control input is constrained such that $\|\bm{u}(t)\| \in \mathcal{U}$ for all times $t\ge 0$. For the simulations shown in Fig. \ref{fig:constrainedinput}, we set the maximum norm of the control input to be $\bar{u} = \SI{1}{m/s^2}$. The black dashed lines represent the inflated boundary of the obstacles depending on the speed of the agent. It should be noted that the obstacle in the top-right quadrant is not inflated since the agent, located at $\bm{P} = [0.5,0.5]^\top$ in the figure, is not in its area of influence. We can notice that the agent is correctly reaching his target while his control input has been capped twice (the second time being when the agent enters the inflated region of the left obstacle, while still avoiding collision). 

\begin{figure}[H]
  \centering
  \includegraphics[width=1\linewidth]{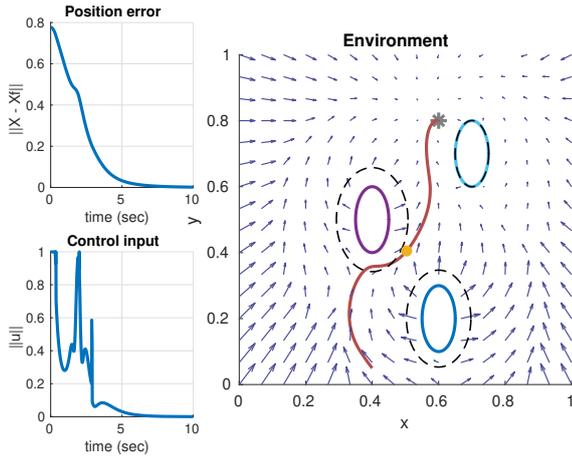}
  \caption{2D obstacle avoidance with constrained input control ($\bm{P}_0 = [-8,0]^\top$ and $\bm{P}_f = [8,0]^\top$).}
  \label{fig:constrainedinput}
\end{figure}

\section{Conclusion} \label{sec:ccl}

In summary, this paper addresses the problem of moving an agent in an environment populated with static and dynamic elliptical obstacles whose shape, position and velocity can change with time. The resulting algorithm is based on a vector field designed to provide information about the desired local agent's velocity and steer the system to a desired target under a double integrator dynamics.
Our work has been extended to include bounded environments and limited control inputs. 
%Next, this vector field has been merged with a task allocation algorithm to incorporate multi-agent settings.

Further research can be done to extend the work to obstacles of any convex shape, and then to obstacles with non-convex shapes. Moreover, additional research on vector field-based algorithms can extend this work, whose applicability is limited to agents with double integrator kinematics, to problems with more general dynamics.
%Finally, the task allocation can be improved to avoid the computation of the control cost for each agent and task, hence reducing the computation time and producing an algorithm more suitable for real-time systems.

\bibliography{main}

\begin{thebibliography}{22}
\providecommand{\natexlab}[1]{#1}
\providecommand{\url}[1]{\texttt{#1}}
\providecommand{\urlprefix}{URL }
\expandafter\ifx\csname urlstyle\endcsname\relax
  \providecommand{\doi}[1]{doi:\discretionary{}{}{}#1}\else
  \providecommand{\doi}{doi:\discretionary{}{}{}\begingroup
  \urlstyle{rm}\Url}\fi

\bibitem[{Aurenhammer(1991)}]{aurenhammer1991voronoi}
Aurenhammer, F. (1991).
\newblock Voronoi diagrams—a survey of a fundamental geometric data
  structure.
\newblock \emph{ACM Computing Surveys (CSUR)}, 23(3), 345--405.

\bibitem[{Chen et~al.(2017)Chen, Everett, Liu, and How}]{chen2017socially}
Chen, Y.F., Everett, M., Liu, M., and How, J.P. (2017).
\newblock Socially aware motion planning with deep reinforcement learning.
\newblock In \emph{2017 IEEE/RSJ International Conference on Intelligent Robots
  and Systems (IROS)}, 1343--1350. IEEE.

\bibitem[{Dijkstra et~al.(1959)}]{dijkstra1959note}
Dijkstra, E.W. et~al. (1959).
\newblock A note on two problems in connexion with graphs.
\newblock \emph{Numerische mathematik}, 1(1), 269--271.

\bibitem[{Fiorini and Shiller(1998)}]{fiorini1998motion}
Fiorini, P. and Shiller, Z. (1998).
\newblock Motion planning in dynamic environments using velocity obstacles.
\newblock \emph{The International Journal of Robotics Research}, 17(7),
  760--772.

\bibitem[{Garg and Panagou(2019)}]{garg2019finite}
Garg, K. and Panagou, D. (2019).
\newblock Finite-time estimation and control for multi-aircraft systems under
  wind and dynamic obstacles.
\newblock \emph{Journal of Guidance, Control, and Dynamics}, 42(7), 1489--1505.

\bibitem[{Hart et~al.(1968)Hart, Nilsson, and Raphael}]{hart1968formal}
Hart, P.E., Nilsson, N.J., and Raphael, B. (1968).
\newblock A formal basis for the heuristic determination of minimum cost paths.
\newblock \emph{IEEE transactions on Systems Science and Cybernetics}, 4(2),
  100--107.

\bibitem[{Huber et~al.(2019)Huber, Billard, and Slotine}]{huber2019avoidance}
Huber, L., Billard, A., and Slotine, J.J. (2019).
\newblock Avoidance of convex and concave obstacles with convergence ensured
  through contraction.
\newblock \emph{IEEE Robotics and Automation Letters}, 4(2), 1462--1469.

\bibitem[{Huttenlocher et~al.(1993)Huttenlocher, Kedem, and
  Sharir}]{huttenlocher1993upper}
Huttenlocher, D.P., Kedem, K., and Sharir, M. (1993).
\newblock The upper envelope of voronoi surfaces and its applications.
\newblock \emph{Discrete \& Computational Geometry}, 9(3), 267--291.

\bibitem[{Khansari-Zadeh and Billard(2012)}]{khansari2012dynamical}
Khansari-Zadeh, S.M. and Billard, A. (2012).
\newblock A dynamical system approach to realtime obstacle avoidance.
\newblock \emph{Autonomous Robots}, 32(4), 433--454.

\bibitem[{Koren et~al.(1991)Koren, Borenstein et~al.}]{koren1991potential}
Koren, Y., Borenstein, J., et~al. (1991).
\newblock Potential field methods and their inherent limitations for mobile
  robot navigation.
\newblock In \emph{ICRA}, volume~2, 1398--1404.

\bibitem[{Latombe(1991)}]{Latombe1991}
Latombe, J.C. (1991).
\newblock \emph{Potential Field Methods}, 295--355.
\newblock Springer US, Boston, MA.
\newblock \doi{10.1007/978-1-4615-4022-9_7}.
\newblock \urlprefix\url{https://doi.org/10.1007/978-1-4615-4022-9_7}.

\bibitem[{Long et~al.(2018)Long, Fan, Liao, Liu, Zhang, and
  Pan}]{long2018towards}
Long, P., Fan, T., Liao, X., Liu, W., Zhang, H., and Pan, J. (2018).
\newblock Towards optimally decentralized multi-robot collision avoidance via
  deep reinforcement learning.
\newblock In \emph{2018 IEEE International Conference on Robotics and
  Automation (ICRA)}, 6252--6259. IEEE.

\bibitem[{Marchidan and Bakolas(2020)}]{marchidan2020collision}
Marchidan, A. and Bakolas, E. (2020).
\newblock Collision avoidance for an unmanned aerial vehicle in the presence of
  static and moving obstacles.
\newblock \emph{Journal of Guidance, Control, and Dynamics}, 43(1), 96--110.

\bibitem[{N{\"u}rnberg(2006)}]{nurnberg2006distance}
N{\"u}rnberg, R. (2006).
\newblock Distance from a point to an ellipse.
\newblock \emph{URl: http://wwwf. imperial. ac. uk/\~{} rn/distance2ellipse.
  pdf}.

\bibitem[{Panagou(2014)}]{panagou2014motion}
Panagou, D. (2014).
\newblock Motion planning and collision avoidance using navigation vector
  fields.
\newblock In \emph{2014 IEEE International Conference on Robotics and
  Automation (ICRA)}, 2513--2518. IEEE.

\bibitem[{Panagou(2016)}]{panagou2016distributed}
Panagou, D. (2016).
\newblock A distributed feedback motion planning protocol for multiple unicycle
  agents of different classes.
\newblock \emph{IEEE Transactions on Automatic Control}, 62(3), 1178--1193.

\bibitem[{Snape et~al.(2011)Snape, Van Den~Berg, Guy, and
  Manocha}]{snape2011hybrid}
Snape, J., Van Den~Berg, J., Guy, S.J., and Manocha, D. (2011).
\newblock The hybrid reciprocal velocity obstacle.
\newblock \emph{IEEE Transactions on Robotics}, 27(4), 696--706.

\bibitem[{Van Den~Berg et~al.(2011)Van Den~Berg, Guy, Lin, and
  Manocha}]{van2011reciprocal}
Van Den~Berg, J., Guy, S.J., Lin, M., and Manocha, D. (2011).
\newblock Reciprocal n-body collision avoidance.
\newblock In \emph{Robotics research}, 3--19. Springer.

\bibitem[{Van~den Berg et~al.(2008)Van~den Berg, Lin, and
  Manocha}]{van2008reciprocal}
Van~den Berg, J., Lin, M., and Manocha, D. (2008).
\newblock Reciprocal velocity obstacles for real-time multi-agent navigation.
\newblock In \emph{2008 IEEE International Conference on Robotics and
  Automation}, 1928--1935. IEEE.

\bibitem[{Warren(1990)}]{warren1990multiple}
Warren, C.W. (1990).
\newblock Multiple robot path coordination using artificial potential fields.
\newblock In \emph{Proceedings., IEEE International Conference on Robotics and
  Automation}, 500--505.

\bibitem[{Wilhelm and Clem(2019)}]{wilhelm2019vector}
Wilhelm, J.P. and Clem, G. (2019).
\newblock Vector field uav guidance for path following and obstacle avoidance
  with minimal deviation.
\newblock \emph{Journal of Guidance, Control, and Dynamics}, 42(8), 1848--1856.

\bibitem[{Yang and Zhao(2004)}]{yang2004trajectory}
Yang, H.I. and Zhao, Y.J. (2004).
\newblock Trajectory planning for autonomous aerospace vehicles amid known
  obstacles and conflicts.
\newblock \emph{Journal of Guidance, Control, and Dynamics}, 27(6), 997--1008.

\end{thebibliography}

\end{document}